\documentclass[]{article}
\usepackage{amsmath,amsfonts,amssymb}
\usepackage{amsthm}
\usepackage{mathtools}
\usepackage{cases}
\usepackage[mathfrak]{}
\usepackage[latin1]{inputenc}
\usepackage[T1]{fontenc}
\usepackage{verbatim}
\usepackage{graphicx}
\usepackage{float}
\usepackage{mathrsfs}
\usepackage [all]{xy}
\usepackage{color}
\usepackage{tikz}
\usepackage{comment}
\usepackage{hyperref}
\usepackage{hyphenat}
\usepackage{authblk}
\usepackage{comment}
\usepackage[english]{babel}
\usepackage{amsfonts}
\usepackage{caption}
\usepackage{float}
\usepackage{placeins}
\usepackage{amssymb}
\usepackage[table]{xcolor}
\usepackage[T1]{fontenc}
\usepackage[english]{babel}
\usepackage{algorithm}
\usepackage{algpseudocode}

\usepackage{amsmath}

\newtheorem{theorem}{Theorem}[section]
\newtheorem{lemma}[theorem]{Lemma}

\newtheorem{proposition}[theorem]{Proposition}%

\newtheorem{example}{Example}%
\newtheorem{remark}{Remark}%

\newtheorem{definition}{Definition}

\title{Efficient Polynomial-Time Decoding of Simplicial Anticodes with Near-Optimal Performance}
\author[1]{Antonio Jes\'us Lorite-L\'opez \footnote{Corresponding author: all871@ual.es}}
\author[2]{Daniel Camaz\'on-Portela}
\author[3]{Juan Antonio L\'opez-Ramos}
\affil[1, 2, 3]{Department of Mathematics, University of Almer\'ia, Carretera Sacramento, SN, Almer\'ia, 04120, Spain}
\date{}                     
\begin{document}
  \maketitle



\begin{abstract}
In this work, we propose an efficient decoding algorithm for codes arising from simplicial complexes, a family of binary linear codes for which no decoding method of this type was previously known.

Although the algorithm does not always attain the maximum theoretical error-correcting capability, it provides an explicit bound that can be computed directly from the structure of the complex. Moreover, this bound is asymptotically optimal: the ratio between the guaranteed correcting capability and the theoretical maximum converges to $1$ as the code length increases, under natural assumptions on the dimension of the maximal faces. The correction capability is also presented in specific examples. Finally, we introduce specific families of simplicial complexes where the algorithm successfully reaches this theoretical bound.
\end{abstract}

\section{Introduction}

Since the introduction of linear codes based on simplicial complexes in \cite{chang2018}, this area has attracted growing interest. Several studies have examined their structural properties, presenting families of optimal codes under some certain assumptions and exploring possible generalisations of the original model \cite{Hu2024,HYUN2019135,Hyun2020}. Furthermore, these codes admit alternative descriptions via trace functions, as shown in \cite{DingDing}, which makes them relevant for applications such as secret sharing \cite{DingDing,Mesnaegr2020}.

However, despite the theoretical progress that has been made, there is currently no efficient decoding algorithm, which has limited their applications, particularly in cryptographic contexts where efficient decoding procedures play a central role \cite{Berlekamp1978OnTI}.

A standard decoding technique for linear codes is syndrome decoding. Let
$C\subseteq\mathbb{F}_2^n$ be an $[n,k]$ linear code with parity-check matrix $H$.
Given a received word $r$, its syndrome is computed as
\[
s=Hr^T.
\]

The code $C$ partitions the ambient space $\mathbb{F}_2^n$ into equivalence classes, called \emph{cosets}, under the relation $r_1 \sim r_2 \iff r_1-r_2\in C$. Two vectors belong to the same coset if and only if they have the same syndrome. Each coset contains a minimum-weight representative, called the \emph{coset leader}, and syndrome decoding consists of identifying this representative and subtracting it from the received word.

Since an $[n,k]$ linear code has $2^{n-k}$ distinct syndromes, there are
$2^{n-k}$ different cosets. Determining the corresponding coset leader is
equivalent to solving the \emph{Nearest Codeword Problem}, which is known to be NP-hard \cite{Berlekamp1978OnTI}. Consequently,
syndrome decoding does not, in general, provide an efficient decoding
algorithm.

To address these limitations, we focus on the anticodes naturally induced by the complementary sets of simplicial complexes. Just as a simplicial complex is uniquely determined by its maximal faces, its complementary set is characterized by its minimal faces. This rich combinatorial and geometric structure allows for a particularly compact representation of the underlying codes, a feature that is highly desirable for reducing storage overhead and key sizes in cryptographic protocols.

Motivated by both this computational difficulty and the structural advantages of these objects, this paper exploits the structure of simplicial anticodes to construct an efficient decoding algorithm. Unlike standard decoding algorithms, the proposed procedure recovers the original message directly rather than the transmitted codeword. As a consequence, no additional step is required to recover the message from the decoded codeword, avoiding the need to solve the corresponding linear system. Moreover, if the ordering of columns is known, the algorithm becomes even more efficient.

Although the algorithm does not always attain the maximum theoretical error-correction capability, it provides an explicit and easily computable bound and is able to correct a number of errors close to this optimum. Moreover, we prove that the ratio between the number of correctable errors and the theoretical error-correction capability converges to $1$ as the code length tends to infinity. In this sense, the proposed algorithm is asymptotically optimal. Comparisons on several examples are also included to illustrate its practical performance.

We identify several families of anticodes for which the proposed decoding algorithm always achieves the theoretical error-correction capability under natural assumptions on the associated simplicial complexes, inspired by ideas developed in \cite{SOLOMON1965170}. We also apply the algorithm to several families of classical codes, including simplex codes, first-order Reed-Muller codes, and MacDonald codes over $\mathbb{F}_2$ \cite{MacWilliams1977,ReedMuller1954,macdonal1960}. Viewing these codes as anticodes, we show that the algorithm achieves the maximum error-correcting capability allowed by the minimum distance.

\section{Preliminaries} 
In this paper, we work over the finite field $\mathbb{F}_2$. This section introduces the family of linear codes under consideration and recalls their most fundamental properties. 
\begin{definition}
    Let $m$ be a natural number and let $[m]$ denote the set $\{1,\ldots,m\}$. A simplicial complex $\Delta$ over $[m]$ is a family of subsets $\Delta \subseteq \mathcal{P}([m])$ such that for every $\sigma \in \Delta$ it is verified that all the subsets of $\sigma$ also belong to $\Delta$.
\end{definition}

This structure allows us to describe a simplicial complex $\Delta$ from their maximal faces, that is, those that are not contained in any other face of the complex $\Delta$.

\begin{definition}
   The complementary of a simplicial complex $\Delta$ is defined as
    \[
        \Delta^c := \mathcal{P}([m]) \setminus \Delta.
    \]
 \end{definition}

Just as a simplicial complex is uniquely determined by its maximal faces, its complement is uniquely determined by its minimal faces, namely, those faces that do not properly contain any other face of the complement.

\begin{definition}
    The linear code associated to the complementary complex, denoted by $C_{\Delta^c}$ and named \emph{anticode}, is generated by the incidence matrix $G$ of $\Delta^c$, that is, 
    \[
        C_{\Delta^c}=\{uG^T : u \in \mathbb{F}_2^m\},
    \]
where the columns of $G^T$ are formed by the characteristic vectors $\chi_\sigma$ of $\sigma \in \Delta^c$:
\[
\chi_\sigma(i)=
\begin{cases}
1, & \text{if } i \in \sigma,\\
0, & \text{if } i \notin \sigma.
\end{cases}
\]
\end{definition}
Equivalently, the code $C_\Delta$ can be defined using the characteristic vectors associated with the faces of the simplicial complex $\Delta$.

Although the parameters of these anticodes, particularly their minimum distance, are not known in general, several families have been identified for which they can be determined explicitly. In our previous work \cite{lorite2026}, we showed that, for a simplicial code
$C_\Delta\subseteq\mathbb{F}_2^m$, the weight of a codeword $c_\Delta(u)$ is given by
\[
w(c_\Delta(u)) = \# \left\{ \sigma\in\Delta:
|\sigma \cap\operatorname{supp}(u)\right| \equiv 1 \pmod 2 \}.
\] 

\begin{lemma} \label{lemma1}
    Let $m$ be a natural number and let $\sigma$ be a non-empty subset of $[m]$. There are exactly $2^{m-1}$ subsets of $[m]$ whose intersection with $\sigma$ has an odd cardinality.
\end{lemma}

\begin{proof}
    Let $k$ be the number of elements of $\sigma$. We can choose $2^{n-k}$ subsets that do not contain elements in $\sigma$. Since $\sigma$ is non-empty, there exist $2^{k-1}$ subsets of $\sigma$ with an odd cardinality. So, there are $2^{k-1} \cdot 2^{n-k} = 2^{n-1}$ combinations of subsets of $[n]$ which contain an odd number of elements of $\sigma$.
\end{proof}

Using this lemma, it is easy to see that the maximum possible weight of a codeword is $\sum_{i=1}^s2^{|A_i|-1}$.

In particular, if $\Delta$ is a simplicial complex whose set of maximal faces $\mathcal{F}=\{A_1,\ldots,A_s\}$ satisfies
\[
A_i\setminus\bigcup_{j\in[s]\setminus\{i\}}A_j\neq\varnothing
\]
for every $i\in[s]$, then the maximum weight is attained by choosing one vertex belonging exclusively to each maximal face. Consequently, the associated anticode $C_{\Delta^c}$ has parameters $[2^m-|\Delta|,\; m,\; 2^{m-1}-\sum_{i=1}^s2^{|A_i|-1}]$.

Furthermore, necessary and sufficient conditions for these anticodes to be optimal were established in \cite{Hyun2020}. In this context, this condition can be expressed as
\[
|\Delta|<2T-l(T)-v_2(T),
\]
where $T=\sum_{i=1}^s2^{|A_i|-1}$, $v_2(T)$ denotes the highest exponent $j$ such that $2^j$ divides $T$, and $l(T)$ is the length of the binary representation of $T$.

\section{Main results}
Throughout this section we present the proposed decoding algorithm and studies its computational complexity, its error-correction capability, and its main properties.

\begin{lemma} \label{lemma_parejas}
    Let $m$ be a natural number, let $\Delta$ be a simplicial complex over the vertex set $[m]$, and let $G$ be the incidence matrix of $\Delta^c$. Then, for every column of $G^T$ having a $0$ in position $i$, there exists another column that differs only in the $i$-th coordinate.
\end{lemma}

\begin{proof}
    Let
    \[
        \Delta^c := \mathcal{P}([m]) \setminus \Delta.
    \]

    If $\sigma \in \Delta^c$ and $\sigma \subset \tau$, then, $[m]\setminus \tau \subset [m]\setminus \sigma \in \Delta$. Therefore, $[m]\setminus \tau \in \Delta$ and $\tau \in \Delta^c$. Thus, if $\sigma \in \Delta^c$ and $i \notin \sigma$, then
    \[
        \tau := \sigma \cup \{i\} \in \Delta^c.
    \]
    Consequently, the characteristic vectors $\chi_\sigma$ and $\chi_\tau$ coincide in all coordinates except in the $i$-th one, where
    \[
        \chi_\tau(i)=1
        \quad \text{and} \quad
        \chi_\sigma(i)=0.
    \]
    
    Therefore, for every column of $G^T$ with a $0$ in position $i$, there is another column that differs from it only in that coordinate.
    
\end{proof}

\begin{definition}
Given a simplicial complex $\Delta \subseteq \mathcal{P}([m])$ we define the subcomplex $\Delta_{-i}$ as follows:
\begin{equation*}
\Delta_{-i}= \{\sigma\in\Delta : i\notin\sigma\}
\end{equation*}
\end{definition}

\begin{lemma}
    Let $\Delta$ be a simplicial complex whose set of maximal faces is $\{A_1,...,A_s\}$. The size of $\Delta_{-i}$ is given by
    \[
|\Delta_{-i}| = \sum_{r=1}^s (-1)^{r+1} \sum_{1 \le j_1 < \dots < j_r \le s} 2^{|B_{j_1} \cap \dots \cap B_{j_r}|},
\]
where $B_j = A_j \setminus \{i\}$.
\end{lemma}

\begin{proof}
     It follows from the inclusion-exclusion principle (see \cite{adamaszek2013}).
\end{proof}

\subsection{The Simplicial Decoder}

Let $n$ and $k$ denote the length and the dimension of the code $C_{\Delta^c}$, respectively.

\textbf{Step 0}

Firstly, given the simplicial complex $\Delta \subseteq \mathcal{P}([m])$, we compute the element $i_{0}\in [m]$ for which the associated $0$-dimensional simplex has the shortest possible link length.
\begin{algorithm}[H]
\caption{Presetting}
\textbf{Input:} $(\Delta)$
\begin{algorithmic}[1]   
    \State $i_0 \gets 1$
    \For{$i\in [m]$}
        \If{$|\Delta_{-i}|>|\Delta_{-i_0}|$}
            \State $i_1 \gets i_0$
            \State $i_0 \gets i$
        \EndIf
    \EndFor
    \State $I \gets [m]\setminus\{i_0\}$
\State\Return $I$
\end{algorithmic}
\end{algorithm}

\begin{remark}
The number of simplices of a simplicial complex can be computed using the inclusion-exclusion principle (see \cite{adamaszek2013}).
\end{remark}


\begin{proposition}
Let $G^T$ be the generator matrix of the anticode $C_{\Delta^c}$, and let
$v=uG^T+e$. Then
\[
v_j+v_{j'}=u_i+e_j+e_{j'},
\]
where $j$ and $j'$ are the indices of two columns of $G^T$ that differ only in
the $i$-th coordinate. In particular, if $e_j=e_{j'}=0$, then
\[
v_j+v_{j'}=u_i.
\]
\end{proposition}

\begin{proof}
By Lemma~\ref{lemma_parejas}, every column of $G^T$ having a $0$ in the $i$-th coordinate is paired with a column having a $1$ in that coordinate and coinciding in all the remaining coordinates. Let
$G_{\cdot,j}$ and $G_{\cdot,j'}$ be such a pair.

Then
\[
v_j=\langle u,G_{\cdot,j}\rangle+e_j,
\qquad
v_{j'}=\langle u,G_{\cdot,j'}\rangle+e_{j'}.
\]
Since the two columns differ only in the $i$-th coordinate,
\[
\langle u,G_{\cdot,j'}\rangle
=
\langle u,G_{\cdot,j}\rangle+u_i.
\]
Therefore,
\[
v_j+v_{j'}=u_i+e_j+e_{j'}.
\]

\end{proof}

\subsubsection*{Step 1}

By the previous Proposition, each pair of columns provides an estimate of $u_i$. The recovered value is incorrect only when exactly one of the two bits $v_j$ and $v_{j'}$ contains an error. If both positions are error-free or both contain errors, the correct value of $u_i$ is still obtained. Repeating this procedure for every pair of columns that differ only in the $i$-th coordinate yields multiple estimates of $u_i$. If fewer than half of these pairs are affected by exactly one error, the correct value of $u_i$ is recovered by a majority vote.

\vspace{5mm}
\begin{algorithm}[H] 
\caption{Simplicial Decoder (Part I) }
\textbf{Input:} $(G^T, n, I, v=(v_1, v_2, \dots, v_n))$. 
\begin{algorithmic}
    \For{$i\in I$}
        \State $N_0, N_1 \gets 0$ 
        \For{$j\in [n]$}
            \If{$g_{ij}=0$}
                \State Select $j'$ such that $g_{ij'}=1$ and $g_{lj}=g_{lj'}$ for all $l\neq i$
                \If{$v_j+v_{j'}=1$}
                    \State $N_1 \gets N_1+1$
                \Else
                    \State $N_0 \gets N_0+1$
                \EndIf
            \EndIf
        \EndFor
        \If{$N_1>N_0$}
            \State $u_i\gets 1$
        \Else
            \State $u_i\gets 0$
        \EndIf
    \EndFor
\end{algorithmic}
\end{algorithm}

Once $k-1$ bits of information have been recovered, we can improve the decoding of the remaining bit $u_{i_0}$.

\begin{proposition}
Let $i_0$ be the unique coordinate of $u$ that has not yet been recovered, and
let $G_{\cdot,s}$ be a column of $G^T$ whose $i_0$-th entry is equal to $1$. Then
\[
\hat{v}_s:=v_s-\sum_{j\neq i_0}u_jg_{j,s}=u_{i_0}+e_s.
\]
Consequently, each such column provides an independent estimate of $u_{i_0}$.
\end{proposition}

\begin{proof}
Let $G_{\cdot,s}$ be a column of $G^T$ with $g_{i_0,s}=1$. Since $v=uG^T+e$,
we have
\[
v_s=\langle u,G_{\cdot,s}\rangle+e_s
    =u_{i_0}+\sum_{j\neq i_0}u_jg_{j,s}+e_s.
\]
The values $\{u_j\}_{j\neq i_0}$ are already known, we define
\[
\hat{v}_s
=
v_s-\sum_{j\neq i_0}u_jg_{j,s}.
\]
It follows that $\hat{v}_s=u_{i_0}+e_s$.

\end{proof}
\subsubsection*{Step 2}
The last bit of information is then recovered by choosing the most frequent value among all estimates $\hat{v}_s$. This method is significantly more robust, since it exploits every column of $G^T$ with $g_{i_0,s}=1$ individually, rather than being restricted to specific pairs of columns.



\setcounter{algorithm}{1}
\begin{algorithm}[H] 
\caption{Simplicial Decoder (Part II) }
\textbf{Step 2: Decode $u_{i_0}$}
\begin{algorithmic}[1]  
    \State $N_0, N_1 \gets 0$
    \For{$j\in [n]$}
        \State $S \gets 0$
        \If{$g_{i_0j}\neq 0$}
            \For{$i\in I$}
                \State $S \gets S+g_{ij}u_i$
            \EndFor
            \If{$v_j-S=1$}
                \State $N_1 \gets N_1+1$
            \Else
                \State $N_0 \gets N_0+1$
            \EndIf
        \EndIf
    \EndFor
    \If{$N_1>N_0$}
        \State $u_{i_0}\gets 1$
    \Else
        \State $u_{i_0}\gets 0$
    \EndIf
    \State \Return $u=(u_1,\dots,u_k)$
\end{algorithmic}
\end{algorithm}

\begin{remark}
Although the algorithm does not explicitly use the value of $i_1$,  computing it is still useful, as  the number of errors that the algorithm can correct depends on $|\Delta_{-i_1}|$.
\end{remark}

\begin{example} \label{example:1} 
Consider the simplicial complex $\Delta$ on the vertex set $[4]$ and maximal faces $\mathcal{F}=\{\{1,2\},\{3,4\}\}$.

It is easy to verify that
\[
|\Delta_{-i}|=5,\qquad \text{for every } i\in[4].
\]
The generator matrix of the associated anticode can be constructed in this way
\[
\begin{array}{c|ccccccccc}
 & \{1,3\} & \{1,4\} & \{2,3\} & \{2,4\} &
 \{1,2,3\} & \{1,2,4\} & \{1,3,4\} & \{2,3,4\} & \{1,2,3,4\} \\
\hline
1 & 1 & 1 & 0 & 0 & 1 & 1 & 1 & 0 & 1 \\
2 & 0 & 0 & 1 & 1 & 1 & 1 & 0 & 1 & 1 \\
3 & 1 & 0 & 1 & 0 & 1 & 0 & 1 & 1 & 1 \\
4 & 0 & 1 & 0 & 1 & 0 & 1 & 1 & 1 & 1
\end{array}
\]
and therefore
\[
G^{T}=
\begin{pmatrix}
1 & 1 & 0 & 0 & 1 & 1 & 1 & 0 & 1\\
0 & 0 & 1 & 1 & 1 & 1 & 0 & 1 & 1\\
1 & 0 & 1 & 0 & 1 & 0 & 1 & 1 & 1\\
0 & 1 & 0 & 1 & 0 & 1 & 1 & 1 & 1
\end{pmatrix}.
\]

Let the information word be $u=(1,1,1,1).$ Its corresponding codeword is
\[
c_{\Delta^c}(u)=uG^T=(0,0,0,0,1,1,1,1,0).
\]
Now suppose that the error vector is $e=(0,0,1,0,0,0,0,0,0)$,
so that the received word is
\[
v=c_{\Delta^c}(u)+e=(0,0,1,0,1,1,1,1,0).
\]

We illustrate the decoding procedure for the first information bit. We first identify the columns of $G^T$ whose first entry is $0$. These are columns $3$, $4$, and $8$. Each of them is paired with its associated column, giving the pairs
\[
(3,5),\qquad (4,6),\qquad (8,9).
\]

For each pair, we compare the corresponding entries of the received word. In the first pair the two entries coincide, so this pair votes for the value $0$. In contrast, the entries differ in the remaining two pairs, and therefore both vote for the value $1$. By majority voting, the first information bit is decoded as $1$. Repeating the same procedure for the remaining information bits allow us to recover the original information vector.
\end{example}

\begin{remark}
    If an additional error were introduced at position $9$, the majority vote for the first information bit would change: two pairs would vote for $0$, while only one pair would vote for $1$. As a result, the algorithm would decode the first bit incorrectly.
\end{remark}

\subsection{Algorithm Complexity}

\begin{proposition}
The decoding algorithm described above runs in time
\[
O((k-1) n \log n).
\]
In particular, its complexity is polynomial in the code length.
\end{proposition}

\begin{proof}
The complexity of the algorithm is determined by two distinct phases:

\begin{enumerate}
    \item \textbf{Recovery of the first $k-1$ bits:} For each bit $u_i$ (where $i \neq i_0$), the algorithm scans the columns of $G^T$ where the $i$-th entry is $0$. There are at most $n$ such columns. For each column, identifying its pair via binary search costs $\mathcal{O}(\log n)$. Thus, recovering $k-1$ bits takes $\mathcal{O}((k-1) n \log n)$.
    
    \item \textbf{Recovery of the last bit:} To recover the bit $u_{i_0}$, the algorithm considers all columns $G_{\cdot,s}$ where the value $g_{i_0,s}$ is $1$ (at most $n$ columns). For each column, the value $\sum_{j \neq i_0} u_j g_{j,s}$ can be computed in $\mathcal{O}(k)$ time. Processing all such columns takes $\mathcal{O}(nk)$ time.
\end{enumerate}

By combining both phases, the overall running time is $\mathcal{O}((k-1)n\log n)$.

\end{proof}

\begin{remark}
If the indexing of the columns is known in advance, locating the required pairs takes $\mathcal{O}(1)$ time, reducing the overall complexity to $\mathcal{O}((k-1)n)$. In particular, if the generator matrix is constructed in standard lexicographical order, as illustrated in Example~\ref{example:1}, the indexing is immediately available, allowing the required column pairs to be located in constant time during decoding. This advantage is analogous to that of a generator matrix in standard form in classical coding theory. Hence, the decoding algorithm runs in polynomial time with respect to the code length.
\end{remark}

\begin{remark}
    A parallelisation process can reduce the complexity to $\mathcal{O}(\log n)$.
\end{remark}


    


\subsection{Error Correction Capacity Analysis}

\begin{lemma} \label{lemma:zerosposition}
Let $i\in[m]$. The number of columns of $G^T$ having a $0$ in the $i$-th coordinate is
\[
2^{m-1}-|\Delta_{-i}|.
\]
\end{lemma}

\begin{proof}
Let $C_i$ denote the set of subsets corresponding to the columns of $G^T$ with a $0$ in the $i$-th coordinate. By the definition of $G$, these are exactly the faces of $\Delta^c$ that do not contain $i$,
\[
C_i = \{ \sigma \in \Delta^c : i \notin \sigma \}.
\]
We can rewrite this set as
\[
C_i = \mathcal{P}([m] \setminus \{i\}) \setminus \{ \sigma \in \Delta : i \notin \sigma \},
\]
Recognizing that $\{ \sigma \in \Delta : i \notin \sigma \}$ is exactly the subcomplex $\Delta_{-i}$, we obtain:
\[
|C_i| = |\mathcal{P}([m] \setminus \{i\})| - |\Delta_{-i}| = 2^{m-1} - |\Delta_{-i}|.
\]
\end{proof}

\begin{theorem} \label{the1}
Let $\Delta$ be a simplicial complex on $[m]$, and let  $C_{\Delta^c} = \{u G^T : u \in \mathbb{F}_2^m\}$ be the associated anticode, where $G$ is the incidence matrix of $\Delta^c$. By considering the decoding procedure described above, the algorithm corrects every error vector of weight
\[
t < \frac{2^{m-1} - |\Delta_{-i}|}{2} \quad \text{for all } i \in [m] \setminus \{i_0\},
\]
where $\Delta_{-i} = \{\sigma \in \Delta : i \notin \sigma\}$ and $i_0 \in [m]$ is chosen such that 
\[
|\Delta_{-i_0}| \geq |\Delta_{-i}| \quad \forall i \in [m].
\]
In particular, the algorithm's error-correction capability is given by
\[
t = \left\lfloor \frac{2^{m-1} - |\Delta_{-i_1}| - 1}{2} \right\rfloor,
\]
where
\[
|\Delta_{-i_1}| = \max\{ |\Delta_{-i}| : i \in [m] \setminus \{i_0\} \}.
\]
\end{theorem}

\begin{proof}
By Lemma \ref{lemma:zerosposition}, the number of columns with $0$ in position $i$ equals $2^{m-1} - |\Delta_{-i}|$. Consequently, for the first $k-1$ bits, the algorithm corrects every error vector of weight:
\[
t < \frac{2^{m-1} - |\Delta_{-i}|}{2}.
\]
Let $\alpha_i$ denote the number of columns, $c$, in the generator matrix $G^T$ with $1$ in position $i$, that is, $\alpha_{i} = \#\{ c \in G^T : c_i = 1 \}$. When decoding the last bit, we can correct up to $\lfloor \frac{\alpha_i - 1}{2} \rfloor$ errors. Let $n$ denote the length of the code and $d$ its minimum distance. Since the minimum distance of these codes never exceeds half of the length (i.e., $d \leq n/2$) and we always have $\alpha_i \geq n/2$, it follows that $\alpha_i \geq d$. Consequently, this decoding step can consistently outperform the standard $\lfloor \frac{d - 1}{2} \rfloor$ correction bound.

To optimize the algorithm, we choose the element $i_0 \in [m]$ that satisfies $|\Delta_{-i_0}| \geq |\Delta_{-i}|$ for all $i \in [m]$ (this means that $u_{i_{0}}$ corresponds to the "weakest" bit with the fewest available pairs) and decode it last. This determines the overall error-correction capability of the algorithm:
\[
t = \left\lfloor \frac{2^{m-1} - |\Delta_{-i_1}| - 1}{2} \right\rfloor,
\]
where $|\Delta_{-i_1}| = \max \{ |\Delta_{-i}| : i \in [m] \setminus \{i_0\} \}$, concluding the proof.
\end{proof}

\begin{example} 
Let $\Delta$ be the simplicial complex on the vertex set $[6]$ whose maximal faces are
\[
\mathcal{F}=\{\{1,2,3\},\{3,4\},\{4,5,6\}\}.
\]
Since removing either vertex $1$ or vertex $6$ produces isomorphic complexes, we have
\[
|\Delta_{-i_0}|=|\Delta_{-i_1}|=12.
\]
Hence, the error-correction capability guaranteed by the algorithm is
\[
t=\left\lfloor\frac{2^{6-1}-12-1}{2}\right\rfloor=9.
\]
On the other hand, the theoretical error-correction capability of the associated code is $11$.

Consider the information vector $u=(1,1,0,1,0,1)$, and let $c_{\Delta^c}(u)$ be its corresponding codeword. If errors are introduced at positions $\{0,1,6,7,11,12,13,15,16\}$,
 the algorithm successfully recovers the original information vector. However, if two additional errors are introduced at positions $2$ and $14$, the algorithm fails and returns $(1,0,1,1,0,1)$ instead of $u$.

Notice that the bound $t=9$ is only a guaranteed correction capability. In fact, for some error patterns with more weight, the algorithm still recovers the transmitted word successfully. For instance, it correctly decodes errors at positions $\{3, 5, 8, 11, 15, 16, 17, 21,24,26,27\}$.

\end{example}

The following result provides a criterion for determining whether decoding has been successful, even when the number of errors exceeds the error-correction capability.

\begin{proposition}
    Let $\Delta$ be a simplicial complex and let $C_{\Delta^c}$ be the associated anticode with parameters $[n,k,d]$. Let $t$ be the error correction capability of the algorithm, such that $t < \lfloor \frac{d-1}{2} \rfloor$. Let $v = uG^T + e = c_{\Delta^c}(u) + e$ be the received word. The algorithm corrects an error pattern $e \in \mathbb{F}_2^n$ with $t < w(e) \leq \lfloor \frac{d-1}{2} \rfloor$ if and only if $w(c_{\Delta^c}(D(v)) + v) \leq \lfloor \frac{d-1}{2} \rfloor$, where $D(v)$ denotes the algorithm output. 
\end{proposition}

\begin{proof}
    If the decoding is successful, $D(v) = u$. Therefore, $w(c_{\Delta^c}(D(v)) + v) = w(c_{\Delta^c}(u) + c_{\Delta^c}(u) + e) = w(e) \leq \lfloor \frac{d-1}{2} \rfloor$.

    Otherwise, if $c_{\Delta^c}(D(v)) = v' \neq c_{\Delta^c}(u)$, we have that $w(v' + c_{\Delta^c}(u)) \geq d$. Since $v = c_{\Delta^c}(u) + e$, we obtain $w(v' + v + e) \geq d$. Therefore, $w(v' + v)  > \lfloor \frac{d-1}{2} \rfloor$.
\end{proof}

\subsubsection{Comparison between theoretical and practical error correction}

Although the algorithm does not always achieve the maximum theoretical correction capability (i.e., $t < (d-1)/2$), it provides an explicit, easily calculable upper bound on the number of correctable errors. Furthermore, as the code length tends to infinity, the capability guaranteed by the algorithm converges to the theoretical correction capability, provided that the dimension of the maximal sets remains bounded.

\begin{theorem} \label{the2}
 Let $C_{\Delta^c}^n$ be the anticode of length $n$ constructed from the simplicial complex $\Delta=\langle A_1,...,A_s\rangle$ on the vertex set $[k]$, $t_n$ denote the error-correction capability guaranteed by the algorithm, and $t_n^{\mathrm{theo}}=\left\lfloor \frac{d_n-1}{2} \right\rfloor$ be the maximum theoretical error-correction capability of $C_{\Delta^c}^n$, where $d_n$ is its minimum distance. 
Assume that the dimension of the maximal sets remains bounded. Then
\[
    \lim_{n\to\infty}  \frac{t_n}{t_n^{\mathrm{theo}}} = 1.
\]
\end{theorem}

\begin{proof}
Let
\[
    \gamma=\max\{2^{|A_i|}: i\in [s]\}.
\]
Since the dimension of the maximal sets is bounded, the constant $\gamma$ does not depend on $n$. Then
\[
    |\Delta_{-i_0}| \leq \sum_{i=1}^s 2^{|A_i|} <(\gamma+1)s,
\]
for all $i_0 \in [k]$. Therefore, the correcting capability guaranteed by the algorithm satisfies
\[
    t_n \geq \frac{2^{n-1}-(\gamma+1)s}{2}.
\]

The distance of a code of this type satisfies $d_n \le 2^{n-1}$. 
This bound follows from Lemma \ref{lemma1} and from the fact that the construction of the anticode involves removing from its generator matrix the columns corresponding to the vectors of the simplicial complex.

Consequently,
\[
    t_n^{\mathrm{theo}} =\left\lfloor \frac{d_n-1}{2} \right\rfloor \le 2^{n-2}-1.
\]

Hence,
\[
    \lim_{n\to\infty}\frac{t_n}{t_n^{\mathrm{theo}}}\geq \lim_{n\to\infty} \frac{t_n}{2^{n-2}-1} \ge \lim_{n\to\infty}\frac{2^{n-1}-(\gamma+1)s}{2(2^{n-2}-1)} =1.
\]
\end{proof}

To illustrate the previous theorem, we carried out a computational simulation. For each value of $k\in\{5,\ldots,11\}$, a set of 50 simplicial complexes with randomly generated maximal faces of bounded dimension was computationally generated. For each individual complex, we computed both the theoretical error-correction capability and the actual error-correction capability achieved by the algorithm. The results were then averaged for each dimension. Table~\ref{tab:simulation} presents these average values, while Figure~\ref{fig:simulation} depicts the convergence of the average ratio $t^{\mathrm{real}}/t^{\mathrm{theo}}$ towards~$1$ as the number of vertices increases.

\begin{table}[h]
\centering
\caption{Average error-correction capabilities obtained from the simulation.}
\label{tab:simulation}
\begin{tabular}{|c|c|c|c|}
\hline
\rowcolor{gray!30}
Dimension &
Average $t^{\mathrm{real}}$ &
Average $t^{\mathrm{theo}}$ &
Average ratio ${t^{\mathrm{real}}}/{t^{\mathrm{theo}}}$\\
\hline
7  & 8.82    & 10.94   & 0.8047\\
8  & 32.46   & 37.66   & 0.8593\\
9  & 87.56   & 96.20   & 0.9096\\
10 & 208.06  & 220.54  & 0.9432\\
11 & 453.00  & 469.38  & 0.9650\\
12 & 960.54  & 979.62  & 0.9805\\
13 & 1974.00 & 1997.28 & 0.9883\\
14 & 4014.30 & 4041.06 & 0.9934\\
\hline
\end{tabular}
\end{table}

\begin{figure}[H]
    \centering
    \includegraphics[width=1\textwidth]{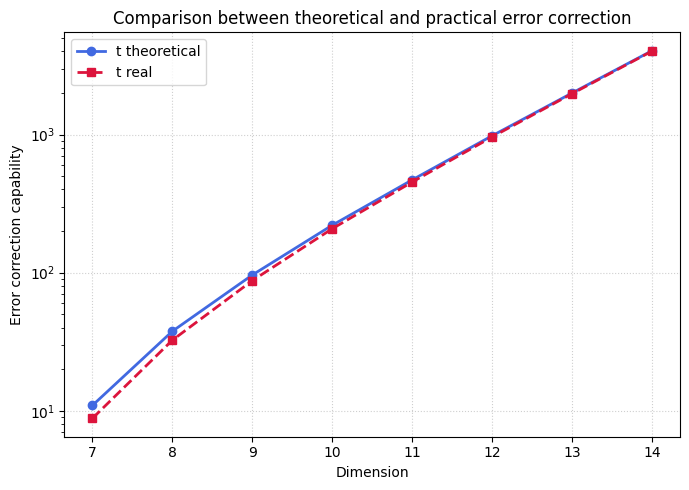}
    \caption{ The ratio approaches $1$ as the number of vertices increases, as predicted in the Theorem~\ref{the2}.}
    \label{fig:simulation}
\end{figure}

\textit{Note:} Although convergence occurs when the code length increases, there may be cases where the theoretical and practical error correction capabilities coincide for short lengths. This can be seen in the table below.

\begin{table}[h]
\centering
\caption{Some specific examples of the algorithm error-correcting capability.}
\renewcommand{\arraystretch}{1.15}
\begin{tabular}{|c|c|c|c|c|}
\hline
\rowcolor{gray!30}
{Set of maximal faces} & $[n,k,d]$ & $t^{\mathrm{real}}$ & $t^{\mathrm{theo}}$ &
${t^{\mathrm{real}}}/{t^{\mathrm{theo}}}$\\
\hline
\multicolumn{5}{|c|}{\cellcolor{gray!15}\textbf{Algorithm attains the theoretical bound}}\\
\hline
$\{\{1,2\},\{3,4\}\}$
& $[9,4,4]$ & $1$ & $1$ & $1$\\
\hline
$\{\{1,2,3,4\},\{2,3,4,5\},\{2,3,6\}\}$
& $[36,6,12]$ & $5$ & $5$ & $1$\\
\hline
$\{\{1,2,3,4,5\},\{6\}\}$
& $[31,6,15]$ & $7$ & $7$ & $1$\\
\hline
$\{\{1,2,3,4,5,6\},\{7\}\}$
& $[63,7,31]$ & $15$ & $15$ & $1$\\
\hline
\multicolumn{5}{|c|}{\cellcolor{gray!15}\textbf{Algorithm does not attain the theoretical bound}}\\
\hline
$\{\{1,2,3\},\{3,4,5\}\}$
& $[18,5,8]$ & $2$ & $3$ & $0.666...$\\
\hline
$\{\{1,2,3\},\{3,4\},\{4,5,6\}\}$
& $[48,6,23]$ & $9$ & $11$ & $0.818...$\\
\hline
$\{\{1,2,3\},\{4,5,6\}\}$
& $[49,6,24]$ & $10$ & $11$ & $0.909...$\\
\hline
$\{\{1,2,3\},\{3,4,5\},\{5,6,7\}\}$
& $[108,7,52]$ & $23$ & $25$ & $0.92$\\
\hline
$\{\{1,2\},\{2,3\},\{3,4\},\{4,5\},\{5,6\},\{6,7\},\{7,1\}\}$
& $[113,7,54]$ & $25$ & $26$ & $0.961...$\\
\hline
$\{\{1,2,3,4,5\},\{6,7,8,9,10\}\}$
& $[961,10,480]$ & $232$ & $239$ & $0.970...$\\
\hline
\end{tabular}
\end{table}
\FloatBarrier

\subsection{Optimal Error-Correction on Certain Families of Simplicial Complexes}

Although a complete characterization of the simplicial complexes for which the algorithm attains the theoretical error-correction capacity remains unknown, our geometric approach allows us to identify several families for which this optimal performance can be proved. Such a characterization is challenging because the minimum distance of the associated anticode is not known in general. In this section, we therefore restrict our attention to simplicial complexes whose set of maximal faces $\{A_1,\ldots,A_s\}$ satisfies $A_i\setminus\bigcup_{j\in[s]\setminus\{i\}}A_j\neq\varnothing$ for every $i\in[s]$; that is, each maximal face contains at least one vertex that does not belong to any other maximal face. Unless otherwise stated, we consider the codes over the set of vertices formed by the union of the maximal faces $\bigcup_{i \in [s]}A_i$.

\begin{theorem} 
Let $C_{\Delta^c}$ be the anticode associated with a simplicial complex of type $\Delta=\langle A_1, A_2\rangle$, where $A_1$ and $A_2$ are maximal faces such that $\dim(A_1)=k-1$, where $k>2$, $\dim(A_2)=1$, and $A_1\cap A_2=\varnothing $. Then, the algorithm attains the theoretical error-correction capability.
\end{theorem}

\begin{proof}
  
 We obtain $|\Delta_{-i_0}|=2^{k-1}$ by deleting the vertex of $A_2$. Taking $i_1$ as a vertex contained only in $A_1$, then
    \[
        |\Delta_{-i_1}|=2^{k-2}+1.
    \]

    We can compute the distance using the formula in \cite{Hyun2020} as
    \[
      d=2^{k-1}-\sum_{i=1}^2 2^{|A_i|-1}=2^{k-1}-1.
    \]
    Using Theorem~\ref{the1}, we obtain that the error-correction capability is 
    \[
        t=\left \lfloor  \frac{2^{k-1}-|\Delta_{-i_1}|-1}{2} \right \rfloor= \left \lfloor  \frac{2^{k-1}-2^{k-2}-1-1}{2} \right \rfloor=\left \lfloor  \frac{d-1}{2} \right \rfloor=t^{theo}.
    \]
\end{proof}

\begin{theorem}
    Let $s\geq 3$ and let $C_{\Delta^c}$ be the anticode associated with a simplicial complex formed by $s$ $\alpha$-simplices sharing a common $(\alpha-1)$-face. Then, the algorithm attains the theoretical error-correction capability.
\end{theorem}

\begin{proof}
Let $A_1,\ldots,A_s$ denote the $\alpha$-simplices.

The dimension is $k=\alpha-1+s$. The theoretical distance is equal to $\sum_{i=1}^s 2^{\alpha-1}$ provided that 
    \[
        \sum_{i=1}^s 2^{\alpha}< 2^{s+(\alpha-1)} \iff s\geq 3.
    \]

Now let $i_0,i_*\in A_i$, where $i_*$ also belongs to some $A_j$ with
$j\neq i$, while $i_0$ does not belong to any other maximal simplex.
Then
\[
|\Delta_{-i_*}|\le |\Delta_{-i_0}|.
\]
Hence, the maximum of $|\Delta_{-i}|$ is attained at a vertex belonging to
exactly one maximal simplex.

Moreover, if $i_r\in A_r$ and $i_t\in A_t$ are vertices contained in no
other maximal simplex and $|A_r|>|A_t|$, then
\[
|\Delta_{-i_r}|\le |\Delta_{-i_t}|.
\]
Since all maximal simplices have the same dimension $\alpha$, every vertex contained in exactly one maximal simplex yields the same value. Therefore, it suffices to consider any such vertex, say $i_1$.

A straightforward computation gives
\[
    |\Delta|=2^\alpha s-2^{\alpha-1}(s-1)
\]
and
\[
    |\Delta_{-i_1}|=2^{\alpha-1}s.
\]
Therefore,
\[
    t= \left\lfloor \frac{2^{\alpha-1+s}-|\Delta_{-i_1}|-1}{2} \right\rfloor
    = \left\lfloor \frac{d-1}{2} \right\rfloor =t^{theo}.
\]
\end{proof}

\begin{theorem}
    Let $\alpha\geq 1$, $s\geq 1$, and $0\leq \beta<\alpha$. Let $C_{\Delta^c}$ be the anticode associated with a simplicial complex consisting of $s$ $\alpha$-simplices sharing a common $(\alpha-1)$-face, together with a $\beta$-simplex attached to this common face along a $(\beta-1)$-face. Then, the algorithm attains the theoretical error-correction capability.
\end{theorem}

\begin{proof}
   The dimension is $k=(\beta+1)+(\alpha-1)=\alpha+\beta$. The distance satisfies
    \[
      d=2^{\alpha+s-1}-\sum_{i=1}^s 2^{\alpha-1}+2^{\beta-1},
    \]
    since $ \sum_{i=1}^{s+1}2^{|A_i|} =s2^\alpha+2^\beta <2^{s+\alpha}$.

    The inequality holds for every $s\ge1$ whenever $\beta<\alpha$.
    
    We obtain that $|\Delta|=s2^\alpha-(s-1)2^{\alpha-1}+2^\beta-2^{\beta-1}=2^{\alpha-1}(s+1)+2^{\beta-1}$. Using the previous reasoning, we have that 
    \[
        |\Delta_{-i_0}|=2^{\alpha-1}(s+1) \text{ and }
        |\Delta_{-i_1}|=s2^{\alpha-1}+2^{\beta-1}.
    \]
    Therefore, 
     \[
        t=\left \lfloor  \frac{2^{\alpha+s-1}-|\Delta_{-i_1}|-1}{2} \right \rfloor=\left \lfloor  \frac{d-1}{2} \right \rfloor=t^{theo}.
    \]
\end{proof}

\textit{Remark:} If $s=1$, according to \cite{Hyun2020}, the anticode is length-optimal provided that
\[
    |\Delta| > 2T - l(T) - v_2(T)
\]

In our case,
\[
    |\Delta| = 2^{k-1}+2^{t-1} \quad\text{and}\quad T = 2^{k-2}+2^{t-1}.
\]
Since
\[
    l(T)=k-1\quad \text{and}\quad v_2(T)=t-1,
\]
the condition becomes
\[
    2^{k-1}+2^{t-1} > 2^{k-1}+2^t-(k-1)-(t-1),
\]
which is equivalent to $k+t>2^{t-1}+2$. In particular, the anticode with $t=1$ is length-optimal whenever $k>2$, while the anticodes corresponding to $t=2,3$ are always length-optimal.

\begin{example}
     Let $\Delta=\langle \{1,2,3,4,5\},\{4,5,6\}\rangle$ on the vertex set $[6]$. The associated length-optimal anticode $C_{\Delta^c}$ has parameters $[28,6,12]$. To illustrate the previous theorem, we compute the error-correction capability achieved by the algorithm. Since
    \[
    |\Delta_{-i_1}|=|\Delta_{-1}|=2^4+2^3-2^2=18,
    \]
    we obtain
    \[
    t=\left\lfloor\frac{2^5-|\Delta_{-i_1}|-1}{2}\right\rfloor
    =\left\lfloor\frac{12-1}{2}\right\rfloor
    =5=t_{\mathrm{th}}.
    \]
    Hence, the algorithm attains the theoretical error-correction capability.
\end{example}

We also present the construction of several families of classical codes based on anticodes. In each case, the error-correcting capability attains the maximum allowed by the minimum distance.
\begin{definition}
    A simplex code \cite{MacWilliams1977} over $\mathbb{F}_2$ of dimension $k$ is given by
    \[
    C_S=\{uH: u \in \mathbb{F}_2^k\},
    \]
    where $H\in \mathcal{M}_{k\times (2^k-1)}$ whose columns are the $2^k-1$ non-zero $k$-tuples. A simplex code has parameters $[2^{k}-1, k, 2^{k-1}]$
\end{definition}
We can obtain an equivalent code using anticodes by considering the simplicial complex $\Delta = \{ \varnothing \}$ on the vertex set $[k]$.

\begin{proposition}
    The algorithm reaches the theoretical error correction capability in simplex codes.
\end{proposition}

\begin{proof}
    Let $\Delta = \{ \varnothing \}$ over $[k]$. Since $|\Delta_{-i}| = 0$ for all $i \in [k]$, the algorithm corrects $t = \left\lfloor \frac{2^{k-1}-1}{2} \right\rfloor = \left\lfloor \frac{d-1}{2} \right\rfloor$ errors.
\end{proof}

\begin{definition}
    A MacDonald code \cite{macdonal1960} $M_{(u, k)}$, for $0 \le u \le k-1$, is a linear code obtained by puncturing the simplex code with dimension $k$. Its generator matrix is derived from the generator matrix $H$ by deleting a specific set of $2^u-1$ columns corresponding to an affine subspace of dimension $u$ over $\mathbb{F}_2^k$. This code has parameters $[2^k - 2^u, k, 2^{k-1} - 2^{u-1}]$.
\end{definition}

We can obtain a code equivalent to $M_{(u,k)}$ through the anticode $C_\Delta^c$ over the vertex set $[k]$ with the associated simplicial complex $\Delta = \langle A_1 \rangle$, where $|A_1| = u$.

\begin{proposition}
    The algorithm reaches the theoretical error correction capability in MacDonald $M_{(u,k)}$ codes.
\end{proposition}

\begin{proof}
    Taking a vertex $i \in A_1$ we obtain that $|\Delta_{-i}|=2^{u-1}$. So the algorithm correction capability is 
    \[
        t=\left\lfloor  \frac{2^{k-1}-2^{u-1}-1}{2} \right\rfloor=\left\lfloor  \frac{d-1}{2} \right\rfloor .
    \]
\end{proof}

\begin{example}
Let $k=3$ and $u=1$. Let $\Delta = \langle \{1\} \rangle$ be the simplicial complex on $[3]$. 

The minimal faces of $\Delta^c$ are $\{2\}$ and $\{3\}$; therefore, the generator matrix of $C_{\Delta^c}$ is 
\[
G=
\left(
\begin{array}{cccccc}
0 & 0 & 0 & 1 & 1 & 1\\
0 & 1 & 1 & 0 & 1 & 1\\
1 & 0 & 1 & 1 & 0 & 1
\end{array}
\right).
\]
The anticode $C_{\Delta^c}$  is equivalent to  $M_{(1, 3)}$ and has parameters $[6, 3, 3]$; so the algorithm corrects $t = \left\lfloor \frac{d-1}{2} \right\rfloor = 1$ error.
\end{example}


As shown above, both the Simplex codes and the MacDonald codes can be regarded as particular instances of the family of codes studied in this paper. In this context, the decoding algorithm developed here avoids the need for pre-computation, unlike previously proposed approaches, which require the construction of PD-sets \cite{Key2016,Washiela2012}. Furthermore, it has a complexity of $\mathcal{O}(n\log n)$. In addition, it can be applied to the entire family of anticodes, a considerably more general class of codes.

\begin{remark}
    The first-order Reed-Muller code \cite{ReedMuller1954} $RM(1, k-1)$ over $\mathbb{F}_2$ of dimension $k$ can be seen as a particular case of a MacDonald code. Specifically, it is equivalent to a simplex code of dimension $k$ from which an affine subspace of dimension $k-1$ has been deleted. It is given by
    \[
     RM(1, k-1) = \{uG : u \in \mathbb{F}_2^{k}\},
    \]
    where 
    \[
         G= \left(
        \begin{array}{c|c}
        \mathbf{0} & H\\
        \hline
        1 & \mathbf{1}
        \end{array}
        \right),
    \] 
    and $H$ denotes the generator matrix of a $(k-1)$-dimensional simplex code. This code has parameters $[2^{k-1}, k, 2^{k-2}]$.
\end{remark}

We can obtain the generator matrix of the first-order Reed-Muller code $RM(1,k-1)$ using the simplicial complex $\Delta$ generated by a single maximal face of $k-1$ vertices on the vertex set $[k]$. 

\begin{example}
    Let $\Delta=\langle\{ 1,2,3\} \rangle$ on $[4]$. The minimal face of $\Delta^c$ is $\{4\}$; therefore, the generator matrix of the anticode $C_{\Delta^c}$ is 
\[
    G=
    \left(
    \begin{array}{c|ccccccc}
    0 & 0 & 0 & 0 & 1 & 1 & 1 & 1\\
    0 & 0 & 1 & 1 & 0 & 0 & 1 & 1\\
    0 & 1 & 0 & 1 & 0 & 1 & 0 & 1\\
    \hline
    1 & 1 & 1 & 1 & 1 & 1 & 1 & 1
    \end{array}
    \right).
\]
\end{example}

The decoding of first-order Reed-Muller codes has recently been studied using variations of permutation decoding, which exploit algebraic factorizations and specific information sets \cite{Bernal2018,Bernal2023}. In contrast, the algorithm proposed in this paper is based on majority-logic decoding. Moreover, it always achieves the maximum theoretical error-correcting capability for this family of codes.

\section*{Declarations}

\textbf{Author contributions.} Antonio Jes\'us Lorite-L\'opez, Daniel Camaz\'on-Portela and Juan Antonio L\'opez-Ramos: Conceptualization, methodology, investigation, analysis, writing-original draft, writing-review and editing. All authors of this article have contributed equally. All authors have read and approved the final version of the manuscript for publication.

 \noindent \textbf{Funding.} This research is supported by PID 2022-140934OB-I00 funded by MI-CIU/AEI/10.13039/501100011033 and ``Junta de Andalucía FQM-425''. The second author was partially supported by grant PID2022-138906NB-C21 funded by MI-CIU/AEI/10.13039/501100011033 and by ERDF A way of making Europe.

\noindent  \textbf{Competing interests.} All authors declare no conflicts of interest in this paper.

\noindent  \textbf{Data and code availability.} The code used to generate the Figure \ref{fig:simulation} is available from the corresponding author upon request.


\bigskip

\bibliographystyle{plain} 
\bibliography{sample} 


\end{document}